\documentclass[a4paper,11pt]{article}
\usepackage{amsthm}
\usepackage{algorithm2e}
\usepackage{color}
\usepackage{graphicx}
\usepackage{makeidx}
\usepackage{amsmath}
\usepackage{amssymb}
\usepackage{float}
\usepackage{amssymb}
\usepackage{float}

\begin{document}

\newtheorem{theorem}{Theorem}
\newtheorem{lemma}{Lemma}
\newtheorem{definition}{Definition}

\title{Gathering with extremely restricted visibility}
\author{Rachid Guerraoui and Alexandre Maurer\\
EPFL}
\date{}
\maketitle

\begin{abstract}
We consider the classical problem of making mobile processes gather or converge at a same position (as performed by swarms of animals in Nature).
Existing works assume that each process can see all other processes, or all processes within a certain radius.

In this paper, we introduce a new model with an extremely restricted visibility: each process can only see \emph{one} other process (its closest neighbor). Our goal is to see if (and to what extent) the gathering and convergence problems can be solved in this setting. We first show that, surprisingly, the problem can be solved for a small number of processes (at most 5), but not beyond. This is due to indeterminacy in the case where there are several ``closest neighbors'' for a same process. By removing this indeterminacy with an additional hypothesis (choosing the closest neighbor according to an order on the positions of processes), we then show that the problem can be solved for any number of processes.
We also show that up to one crash failure can be tolerated for the convergence problem.

\end{abstract}

\section{Introduction}

An interesting natural phenomenon is the ability of swarms of simple individuals to form complex and very regular patterns: swarms of fishes \cite{x1}, birds \cite{x2}, ants \cite{x3}... They do so in a totally distributed manner, without any centralized or irreplaceable leader. Such behaviors are a great source of inspiration for distributed computing.

Problems of \emph{pattern formation} have been extensively studied by the distributed computing community \cite{u6,u7,z1,z2}. In order to prove mathematical results, the model is of course simplified: the individuals (\emph{processes}) are usually geometric points in a Euclidean space. A famous example is the circle formation algorithm by Suzuki and Yamashita \cite{u7}.

In particular, a pattern formation problem which has been extensively studied is the \emph{gathering} problem \cite{u2,u3,u4,b1,b4}: processes must gather at a same point in a finite time. This apparently simple problem can become surprisingly complex, depending on the model and hypotheses: scheduler, symmetry, computational power, memory, orientation... When gathering is impossible, a close problem is the \emph{convergence} problem \cite{u5,b2}: processes must get always closer to a same point.

One of these hypotheses is \emph{visibility}. Most pattern formation papers assume \emph{unlimited} visibility \cite{u1,u2,u3,u4,u5,u6,u7}: processes have a global view of the position of other processes. Some papers assume a \emph{limited} visibility \cite{b1,b2,b3,b4}: processes can only see other processes within a certain radius.

However, even with a limited visibility, each process is supposed to analyze the position of several neighbor processes at each computing step. This leads us to the following question: what is the simplest hypothesis we could make on visibility?

In this paper, we assume that each process can only see its closest neighbor (i.e., the closest other process), and ignores the total number of processes. To our knowledge, no paper has yet considered such a minimalist setting. We then study to what extent the gathering and convergence problems can be solved in this setting. We assume a synchronous scheduler and memoryless processes that cannot communicate with messages.

There is an indeterminacy in the case where there are several ``closest neighbors'' (i.e., two or more processes at the same distance of a given process). We first assume that, in this situation, the closest neighbor is arbitrarily chosen by an external adversary (worst-case scenario).

In this scenario, we show that, surprisingly, the problems can only be solved for a small number of processes. More precisely, if $n$ is the number of processes and $d$ is the number of dimensions of the Euclidean space, then the gathering (resp. convergence) problem can be solved if and only if $d = 1$ or $n \leq 2$ (resp. $d = 1$ or $n \leq 5$). Indeed, for larger values of $n$, there exists initial configurations from which gathering or convergence is impossible, due to symmetry.
The proof is constructive: for the small values of $n$, we provide an algorithm solving the problems. The proof is non-trivial for $n = 4$ and $n = 5$, as several families of cases need to be considered.

Therefore, to solve the problems for larger values of $n$, one additional hypothesis must necessarily be added. We remove the aforementioned indeterminacy by making the choice of the closest neighbor (when there is more than one) deterministic instead of arbitrary (according to an order on the positions of processes). Then, we show that the gathering problem is always solved in at most $n - 1$ steps by a simple ``Move to the Middle'' (MM) algorithm.

We finally consider the case of crash failures,
where at most $f$ processes lose the ability to move.
We show that the gathering (resp. convergence) problem can only be solved when $f = 0$ (resp. $f \leq 1$). When the convergence problem can be solved, the MM algorithm solves it.

Beyond this first work, we believe that this minimalist model can be the ground for many other interesting results.

The paper is organized as follows.
In Section~\ref{secmodel}, we define the model and the problems.
In Section~\ref{secalgo}, we characterize the class of algorithms allowed by our model, and define a simple algorithm to prove the positive results. 
In Section~\ref{seclb}, we prove the aforementioned lower bounds.
In Section~\ref{secpos}, we remove indeterminacy and show that the gathering problem can be solved for any $n$.
In Section~\ref{secft}, we consider the case of crash failures.
We conclude in Section~\ref{sec_conc}.

\section{Model and problems}
\label{secmodel}

\paragraph{Model.} We consider a Euclidean space $S$ of dimension $d$ ($d \geq 1$).
The position of each point of $S$ is described by $d$ coordinates $(x_1,x_2,\dots,x_d)$ in a Cartesian system.
For two points $A$ and $B$ of coordinates $(a_1,\dots,a_d)$ and $(b_1,\dots,b_d)$,
let $d(A,B) = \sqrt{\Sigma_{i=1}^{i=d} (a_i - b_i)^2}$ be the distance between $A$ and $B$.

Let $P$ be a set of $n$ processes. $\forall p \in P$, let $M_p$ be the position of $p$ in $S$. Let $\Omega$ be the set of positions occupied by the processes of $P$. As several processes can share the same position, $1 \leq |\Omega| \leq |P|$.
The time is divided in discrete steps $t \in \{0,1,2,3,\dots\}$.

If $|\Omega| = 1$, the processes are \emph{gathered} (they all have the same position).
If $|\Omega| \geq 2$, $\forall p \in P$, let $D(p) = \min_{K \in \Omega - \{M_p\}} d(M_p,K)$,
and let $N(p)$ be the set of processes $q$ such that $d(M_p,M_q) = D(p)$. At a given time $t$, the \emph{closest neighbor} of a process $p$ is a process of $N_p$ arbitrarily chosen by an external adversary. We denote it by $C(p)$.

We consider a synchronous execution model. 
At a given time $t$, a process $p$ can only see $M_p$ and $M_{C(p)}$ (without global orientation), and use these two points to compute a new position $K$. Then, the position of $p$ at time $t+1$ is $K$.

The processes are oblivious (they have no memory), mute (they cannot communicate) and anonymous (they cannot distinguish each other with identifiers). Note that this model does not assume multiplicity detection (the ability to count the processes at a same position). The processes do not know $n$. At $t = 0$, the $n$ processes can have any arbitrary positions.

\paragraph{Problems.} For a given point $G \in S$ and a given constant $\epsilon$, we say that the processes are $(G,\epsilon)$-gathered if, $\forall M \in \Omega$,
$d(G,M) \leq \epsilon$.

An algorithm solves the \emph{convergence} problem if, for any initial configuration, there exists a point $G \in S$ such that, $\forall \epsilon > 0$, there exists a time $T$ such that the processes are $(G,\epsilon)$-gathered $\forall t \geq T$.

An algorithm solves the \emph{gathering} problem if, for any initial configuration, there exists a point $G$ and a time $T$ such that the processes are $(G,0)$-gathered $\forall t \geq T$.

\section{Algorithm}
\label{secalgo}

In this section, we describe all possible algorithms that our model allows. Doing so enables us to show lower bounds further -- that is, showing that \emph{no algorithm} can solve some problems in our model.
This is not to confuse with the $MM$ algorithm (a particular case, defined below), which is only used to prove positive results.

Here, an algorithm consists in determining, for any process $p$, the position of $p$ at the next step, as a function of $M_p$ and $M_{C(p)}$.

First, let us notice that, if the processes are gathered ($|\Omega| = 1$), the processes have no interest in moving anymore.
This corresponds to the case where each process cannot see any ``closest neighbor''. Thus, we assume that any algorithm is such that, when a process $p$ cannot see any closest neighbor, $p$ does not move.

Now, consider the case where the processes are not gathered $(|\Omega| \geq 2)$. Let $p$ be the current process, let $D = D(p)$, and let $\vec{x}$ be the unit vector ($ || \vec{x} || = 1$)
directed from $M_p$ to $M_{C(p)}$. There are $2$ possible cases.

\paragraph{Case 1: $d = 1$.}

The next position of $p$ is
$M_p + f_x(D)\vec{x}$,
where $f_x$ is an arbitrary function.

\paragraph{Case 2: $d \geq 2$.}

Let $\Delta$ be the axis defined by $M_p$ and $M_{C(p)}$.
If $d \geq 2$, as there is no global orientation of processes ($M_p$ can only position itself relatively to $M_{C(p)}$), the next position of $p$ can only be determined by (1) its position on axis $\Delta$ and (2) its distance to $\Delta$. The difference here is that, for two given parameters (1) and (2), there are several possible positions ($2$ positions for $d = 2$, an infinity of positions for $d \geq 3$).
Thus, we assume that the next position (among these possible positions) is arbitrarily chosen by an external adversary.

More formally, the next position of $p$ is $M_p + f_x(D)\vec{x} + f_y(D)\vec{y}$,
where $f_x$ and $f_y$ are arbitrary functions, and where $\vec{y}$ is a vector orthogonal to $\vec{x}$ which is arbitrarily chosen by an external adversary.

\paragraph{Move to the Middle (MM) algorithm.}

We finally define one particular algorithm to show some upper bounds. The Move to the Middle (MM) algorithm consists, for each process $p$ and at each step, in moving to the middle of the segment defined by $M_p$ and $M_{C(p)}$.

More formally, if $d = 1$, the MM algorithm is defined by $f_x(D) = D/2$.
If $d \geq 2$, the MM algorithm is defined by $f_x(D) = D/2$ and $f_y(D) = 0$.

\section{Lower bounds}
\label{seclb}

In this section, we show the two following results.

\begin{theorem}
The gathering problem can be solved if and only if $d = 1$ or $n \leq 2$. When it can be solved, the MM algorithm solves it.
\label{th_lb_gath}
\end{theorem}

\begin{theorem}
The convergence problem can be solved if and only if $d = 1$ or $n \leq 5$. When it can be solved, the MM algorithm solves it.
\label{th_lb_conv}
\end{theorem}

\subsection{Gathering problem}

Let us prove Theorem~\ref{th_lb_gath}.

\begin{lemma}
\label{lemd1}
If $d = 1$, the MM algorithm solves the gathering problem.
\end{lemma}

\begin{proof}
Let us show that, if $|\Omega| \geq 2$, then $|\Omega|$ decreases at the next step.

As $d=1$, let $x(K)$ be the coordinate of point $K$.
Let $(K_1,K_2,\dots,K_m)$ be the points of $\Omega$ ranked such that $x(K_1) < x(K_2) < \dots < x(K_m)$.
$\forall i \in \{1,\dots,m\}$, let $x_i = x(K_i)$.
Then, according to the MM algorithm, the possible positions at the next step are:
$(x_1+x_2)/2, (x_2+x_3)/2, \dots, (x_{m-1}+x_m)/2$
(at most $m - 1$ positions). Thus, $|\Omega|$ decreases at the next step.
Therefore, after at most $n-1$ steps, we have $|\Omega| = 1$, and the gathering problem is solved.\end{proof}

\begin{lemma}
\label{lem3pt}
If $d \geq 2$ and $n \geq 3$, the gathering problem is impossible to solve.
\end{lemma}

\begin{proof}
First, consider the case $d = 2$.
Consider an initial configuration where $\Omega$ contains three distinct points $K_1$, $K_2$ and $K_3$ such that $d(K_1,K_2) = d(K_2,K_3) = d(K_3,K_1) = D$.

Let $G$ be the gravity center of the triangle $K_1K_2K_3$.
Let $s(1) = 2$, $s(2) = 3$ and $s(3) = 1$.
$\forall i \in \{1,2,3\}$,
let $A_i$ and $B_i$ be the two half-planes delimited by the axis $(K_iK_{s(i)})$, such that $G$ belongs to $B_i$.
Let $\vec{v_i}$ be the unit vector orthogonal to
$(K_iK_{s(i)})$
such that the point $K_i + \vec{v_i}$
belongs to $A_i$.
Let $\vec{y_i} = \vec{v_i}$ if
$f_y(D) \geq 0$, and
$\vec{y_i} = -\vec{v_i}$ otherwise.

Let $p$ be a process, and let $i$ be such that $M_p = K_i$.
The external adversary can choose a closest neighbor $C(p)$ and a vector
$\vec{y}$ such that
$M_{C(p)} = K_{s(i)}$ and
$\vec{y} = \vec{y_i}$.

Thus, at the next step, it is always possible that $\Omega$ contains three \emph{distinct} points also forming an equilateral triangle. The choice of vectors
$\vec{y}$ prevents the particular case where all processes are gathered in point $G$.
We can repeat this reasoning endlessly. Thus, the gathering problem cannot be solved if $d = 2$.

Now, consider the case $d > 2$.
The external adversary can choose the
$\vec{y}$ vectors such that the points of
$\Omega$ always remain in the same plane, and their behavior is the same as for $d = 2$. Thus, the gathering problem cannot be solved if $d > 2$.
\end{proof}

\textbf{Theorem~\ref{th_lb_gath}.} \emph{The gathering problem can be solved if and only if $d = 1$ or $n \leq 2$. When it can be solved, the MM algorithm solves it.}

\begin{proof}
If $d = 1$, according to Lemma~\ref{lemd1}, the MM algorithm solves the gathering problem.
If $n = 1$, the gathering problem is already solved by definition.
If $n = 2$, the MM algorithm solves the gathering problem in at most one step.
Otherwise, if $d \geq 2$ and $n \geq 3$,
according to Lemma~\ref{lem3pt},
the gathering problem cannot be solved.
\end{proof}

\subsection{Convergence problem}

Let us prove Theorem~\ref{th_lb_conv}.

We first introduce some definitions. For a given set of points $X \subseteq S$, let $D_{\max}(X) =  \max_{\{A,B\} \subseteq X}$ $d(A,B)$.
Let $\Omega(t)$ be the set $\Omega$ at time $t$.

Let $d_{\max}(t) = \max_{\{A,B\} \subseteq \Omega(t)} d(A,B)$ and
$d_{\min}(t) = \min_{\{A,B\} \subseteq \Omega(t)} d(A,B)$.
Let $m(A,B)$ be the middle of segment [AB].
Let $\alpha(K) =$ $\sqrt{1 - 1/(4K^2)}$.
Let $R(t) = \arg \min_{G \in S} \max_{M \in \Omega(t)} d(G,M)$ (the radius of the smallest enclosing ball of all processes' positions).
Let $X_i(t)$ be the smallest $i^{th}$ coordinate of a point of $\Omega(t)$.
We say that a proposition $P(t)$ is true
\emph{infinitely often} if, for any time $t$, there exists a time $t' \geq t$ such that $P(t)$ is true.

\begin{lemma}
\label{lem_of_3}
If there exists a time $t$ such that $|\Omega(t)| \leq 3$, the MM algorithm solves the convergence problem.
\end{lemma}

\begin{proof}

If $|\Omega(t)| = 1$, the processes are and remain gathered. If $|\Omega(t)| = 2$, then $|\Omega(t+1)| = 1$.

If $|\Omega(t)| = 3$, consider the following proposition $P$: there exists $t' > t$ such that $|\Omega(t')| \leq 2$. If $P$ is true, the gathering (and thus, convergence) problem is solved. Now, consider the case where $P$ is false.

Let $\Omega(t) = \{A,B,C\}$.
Then, as $|\Omega(t+1)| = 3$, $\Omega(t+1) = \{m(A,B),m(B,C),m(C,A)\}$.
The center of gravity $G$ of the triangle formed by the three points of $\Omega$ always remains the same, and $d_{\max}(t)$ is divided by two at each step.
Thus, $\forall \epsilon > 0$, there exists a time $T$ such that the processes are $(G,\epsilon)$-gathered $\forall t \geq T$.\end{proof}

\begin{lemma}
\label{lemKalpha}
Let $K \geq 1$. If $R(t) \leq K d_{\min}(t)$,
then $R(t+1) \leq \alpha(K) R(t)$.
\end{lemma}

\begin{proof}
If the processes move according to the MM algorithm,
then
$\Omega(t+1) \subseteq \bigcup_{\{A,B\} \subseteq \Omega(t)}$ $\{m(A,B)\}$.
Let $G$ be such that, $\forall M \in \Omega(t)$,
$d(G,M) \leq R(t)$.
Let $A$ and $B$ be two points of $S$ such that $d(G,A) = d(G,B) = R(t)$ and $d(A,B) = d_{\min}(t)$ (two such points $A$ and $B$ exist, as $d_{\min}(t) \leq 2R(t)$).
Let $C = m(A,B)$.
Then, $\forall M \in \Omega(t+1)$,
$d(G,M) \leq d(G,C)$.
Thus, $R(t+1) \leq d(G,C)$.

Let $x = d(G,C)$, $y = d_{\min}(t)/2$ and $z = R(t)$.
Then, $z^2 = x^2 + y^2$ and
$x/z = \sqrt{1 - (y/z)^2}$.
As $R(t) \leq K d_{\min}(t)$, $y/z \geq 1/(2K)$
and $x/z \leq \sqrt{1 - 1/(4K^2)} = \alpha(K)$.
Thus, $R(t+1) \leq d(G,C) \leq \alpha(K) R(t)$.
\end{proof}

\begin{lemma}
\label{lem_milieux}
Let $A$, $B$, $C$, $D$ and $E$ be five points (some of them may be identical). 
Let $x = d(A,D)/100$.
Assume $d(A,B) \leq x$, $d(A,C) \leq x$,
$d(A,E) \leq 100 x$ and $d(D,E) \geq 40 x$.
Let $S = \{A,B,C,D,E\}$ and $S' = \bigcup_{\{A,B\} \subseteq S} \{m(A,B)\}$.
Then, $D_{\max}(S') \leq 0.99 D_{\max}(S)$.
\end{lemma}

\begin{proof}

As $d(A,D) = 100x$, $D_{\max}(S) \geq 100x$.

Let $M_1 = m(A,D)$, $M_2 = m(A,E)$ and $M_3 = m(D,E)$.
We have $d(A,M_1) = 50x$ and $d(A,M_2) \leq 50x$.
The maximal value of $y = d(A,M_3)$ is reached when
$d(A,D) = d(A,E) = 100x$ and $d(D,E) = 40x$.
In this case, with the Pythagorean theorem, we have
$(100x)^2 = y^2 + (20x)^2$,
and thus $y \leq 98x$.

Thus, $\max_{i \in \{1,2,3\}} d(A,M_i) \leq 98x$.
Now, suppose that $D_{\max}(S') > 99x$.
Let $M_4 = m(A,B)$
and $M_5 = m(A,C)$.
This would imply that there exists $i \in \{1,2,3\}$ such that either $d(M_i,M_4) > 99x$ or
$d(M_i,M_5) > 99x$, and thus,
that either $d(A,B) > x$ or $d(A,C) > x$, which is not the case.
Thus, 
$D_{\max}(S') \leq 99x \leq 0.99 D_{\max}(S)$.\end{proof}

\begin{lemma}
\label{lem3A}
Let $t$ be a given time.
If $n = 5$ and $|\Omega(t)| = 5$,
then one of the following propositions is true:
\\(1) $|\Omega(t+1)| \leq 4$
\\(2) $R(t+1) \leq \alpha(1000) R(t)$
\\(3) $d_{\max}(t+1) \leq 0.99 d_{\max}(t)$
\end{lemma}

\begin{proof}
Suppose that (1) and (2) are false.
According to Lemma~\ref{lemKalpha},
(2) being false implies that $R(t) > 1000 d_{\min}(t)$.
Let $A_0$ and $B_0$ be two points of $\Omega(t)$
such that $d(A_0,B_0)$ $= d_{\min}(t)$.
As $|\Omega(t+1)| = 5$, it implies that the processes at $A_0$ and $B_0$ did not both move to $m(A_0,B_0)$.
Therefore, there is a point $C$ of $\Omega(t)$ such that $d(A_0,C)$ $= d_{\min}(t)$ or $d(B_0,C) = d_{\min}(t)$.
If $d(A_0,C) = d_{\min}(t)$, let $A = A_0$ and $B = B_0$. Otherwise, let $A = B_0$ and $B = A_0$.

As $R(t) > 1000 d_{\min}(t)$,
there exists a point $D_0$ of $\Omega(t)$ such that $d(A,D_0) \geq 100 d_{\min}(t)$.
Let $E_0$ be the fifth point of $\Omega(t)$.
If $d(A,D_0) \geq d(A_0,E_0)$,
let $D = D_0$ and $E = E_0$. Otherwise, let $D = E_0$ and $E = D_0$.

Finally, let $x = d(A,D)/100$.
Thus, we have $d(A,B)$ $\leq x$,
$d(A,C) \leq x$ and $d(A,E) \leq 100x$.
If $d(D,E) < 40x$, then the processes at positions $D$ and $E$ both move to $m(D,E)$,
and $|\Omega(t+1)| = 4$: contradiction.
Thus, $d(D,E) \geq 40x$.
Let $S = \Omega(t)$, and let
$S' = \bigcup_{\{A,B\} \subseteq S} \{m(A,B)\}$.
Then, according to Lemma~\ref{lem_milieux},
\\$D_{\max}(S') \leq 0.99 D_{\max}(S)$.

As the processes move according to the MM algorithm,
$\Omega(t+1) \subseteq S'$, and 
$d_{\max}(t+1) \leq D_{\max}(S') \leq 0.99 D_{\max}(S) = 0.99 d_{\max}(t)$.
Thus, (3) is true.

Therefore, either (1) or (2) are true, or (3) is true.\end{proof}

\begin{lemma}
\label{lem3B}
Let $t$ be a given time.
If $|\Omega(t)| = 4$, then one of the following propositions is true:
\\(1) $|\Omega(t+1)| \leq 3$
\\(2) $R(t+1) \leq \alpha(1000) R(t)$ 
\\(3) $d_{\max}(t+1) \leq 0.99 d_{\max}(t)$
\end{lemma}

\begin{proof}
Suppose that (1) and (2) are false.
According to Lemma~\ref{lemKalpha},
(2) being false implies that $R(t) > 1000 d_{\min}(t)$.
Let $A$ and $B$ be two points of $\Omega(t)$ such that $d(A,B) = d_{\min}(t)$.

As $R(t) > 1000 d_{\min}(t)$,
there exists a point $D_0$ of $\Omega(t)$ such that $d(A,D_0) \geq 100 d_{\min}(t)$.
Let $E_0$ be the fourth point of $\Omega(t)$.
If $d(A,D_0) \geq d(A_0,E_0)$,
let $D = D_0$ and $E = E_0$.
Otherwise, let $D = E_0$ and $E = D_0$.

Let $C = A$ and $x = d(A,D)/100$.
Thus, we have $d(A,B) \leq x$, $d(A,C) \leq x$ and $d(A,E) \leq 100x$.
If $d(D,E) < 40x$, then the processes at $D$ and $E$ (resp. $A$ and $B$) both move to $m(D,E)$ (resp. $m(A,B)$),
and $|\Omega(t+1)| = 2$: contradiction.
Thus, $d(D,E) \geq 40x$.

Let $S = \Omega(t)$, and let
$S' = \bigcup_{\{A,B\} \subseteq S} \{m(A,B)\}$.
Then, according to Lemma~\ref{lem_milieux},
$D_{\max}(S') \leq 0.99 D_{\max}(S)$.

As the processes move according to the MM algorithm,
$|\Omega(t+1)| \subseteq S'$, and 
$d_{\max}(t+1) \leq D_{\max}(S') \leq 0.99 D_{\max}(S) = 0.99 d_{\max}(t)$.
Thus, (3) is true.

Therefore, either (1) or (2) are true, or (3) is true.\end{proof}

\begin{lemma}
\label{lem_r_is_dec}
At any time $t$, $R(t+1) \leq R(t)$.
\end{lemma}

\begin{proof}
Suppose the opposite: $R(t+1) > R(t)$.
Let $G$ be a point such that, $\forall M \in \Omega(t)$, $d(G,M) \leq R(t)$.
If, $\forall M \in \Omega(t+1)$,
$d(G,M) \leq R(t)$, then we do not have $R(t+1) > R(t)$.
Thus, there exists a point $A$ of $\Omega(t+1)$ such that $d(G,A) > R(t)$.
Let $B$ be the previous position of processes at position $A$.
As the processes at position $B$ moved to $A$, according to the MM algorithm,
there exists a point $C$ of $\Omega(t)$ such that $A = m(B,C)$.
As $d(G,B) \leq R(t)$ and $d(G,A) > R(t)$, we have $d(G,C) > R(t)$.
Thus, there exists a point $C$ of $\Omega(t)$ such that $d(G,C) > R(t)$: contradiction.
\end{proof}

\begin{lemma}
\label{lem_dmax_dim}
At any time $t$, $d_{\max}(t+1) \leq d_{\max}(t)$.
\end{lemma}

\begin{proof}

Suppose the opposite: $d_{\max}(t+1) > d_{\max}(t)$.
Let $A$ and $B$ be two points of $\Omega(t+1)$ such that $d(A,B) = d_{\max}(t+1)$.
According to the MM algorithm, there exists four points $A_1$, $A_2$, $B_1$ and $B_2$ of $\Omega(t)$ such that
$A = m(A_1,A_2)$ and $B = m(B_1,B_2)$.

Let $L$ be the line containing $A$ and $B$.
Let $A'_1$ (resp. $A'_2$, $B'_1$ and $B'_2$) be the projection of $A_1$ (resp. $A_2$, $B_1$ and $B_2$) on $L$.
Then, there exists $i \in \{1,2\}$ and $j \in \{1,2\}$ such that
$d(A'_i,B'_j) \geq d(A,B)$.
Thus, $d(A_i,B_j) \geq d(A,B) = d_{\max}(t)$: contradiction.

\end{proof}

\begin{lemma}
\label{lemp1p2}
Let $n \leq 5$.
Let $P_1(t)$ (resp. $P_2(t)$) be the following proposition:
$R(t+1) \leq \alpha(1000) R(t)$
(resp. $d_{\max}(t+1) \leq 0.99 d_{\max}(t)$).
Let $P(t) = P_1(t) \vee P_2(t)$.
If, for any time $t$, $|\Omega(t)| \geq 4$,
then $P(t)$ is true infinitely often.
\end{lemma}

\begin{proof}
Let $P^*$ be the following proposition:
``$|\Omega(t)| = 4$''
is true infinitely often.

If $P^*$ is false, there exists a time $t'$ such that $\forall t \geq t'$,
$|\Omega(t)| = 5$.
Thus, the result follows, according to Lemma~\ref{lem3A}.
If $P^*$ is true, there exists an infinite set $T = \{t_1,t_2,t_3 \dots \}$ such that $\forall t \in T$, 
$|\Omega(t)| = 4$.
Then, according to Lemma \ref{lem3B}, 
$P(t+1)$ is true $\forall t \in T$. Thus, the result follows.
\end{proof}

\begin{lemma}
\label{lem_bloc_1_2}
Let $n \leq 5$.
Suppose that, for any time $t$, $|\Omega(t)| \geq 4$.
Then, for any time $t$, there exists a time $t' > t$ such that $R(t') \leq \alpha(1000) R(t)$.
\end{lemma}

\begin{proof}
Suppose the opposite: there exists a time $t_0 $ such that,
$\forall t > t_0$,
$R(t) > \alpha(1000) R(t_0)$.

Consider the propositions $P_1(t)$ and $P_2(t)$ of
Lemma \ref{lemp1p2}.
Then, $\forall t \geq t_0$, $P_1(t)$ is false.
Thus, according to Lemma \ref{lemp1p2},
it implies that $P_2(t)$ is true infinitely often.

Let $t' > t_0$ be such that, between time $t_0$ and time $t'$, $P_2(t)$ is true at least $200$ times.
According to Lemma~\ref{lem_dmax_dim},
for any time $t$, we have $d_{\max}(t+1) \leq d_{\max}(t)$.
Thus, $d_{\max}(t') \leq 0.99^{200} d_{\max}(t_0) \leq d_{\max}(t_0)/4$.
For any time $t$, $d_{\max}(t) \geq R(t)$ and
$d_{\max}(t) \leq 2R(t)$.
Thus, $R(t') \leq R(t_0)/2 \leq \alpha(1000) R(t_0)$: contradiction.
Thus, the result follows.
\end{proof}

\begin{lemma}
\label{lem_dgm_dec}
Let $G$ be a point such that, $\forall M \in \Omega(t)$,
$d(G,M) \leq R(t)$. 
Then, $\forall M \in \Omega(t+1)$, $d(G,M) \leq R(t)$. 
\end{lemma}

\begin{proof}
Suppose the opposite:
there exists a point $K$ of $\Omega(t+1)$ such that
$d(G,K) > R(t)$. 
According to the MM algorithm, there exists two points $A$ and $B$ of $\Omega(t)$ such that $K = m(A,B)$.
Then, as $d(G,K) > R(t)$, either $d(G,A) > R(t)$ or $d(G,B) > R(t)$: contradiction. Thus, the result follows.
\end{proof}

\begin{lemma}
\label{lemshrink2rt}
$\forall i \in \{1,\dots,d\}$
and for any two instants $t$ and $t' > t$,
$|X_i(t') - X_i(t)| \leq 2R(t)$.
\end{lemma}

\begin{proof}
For any point $M$, let $x_i(M)$ be
the $i^{th}$ coordinate of $M$.
Let $G$ be a point such as described in
Lemma~\ref{lem_dgm_dec}.
According to Lemma~\ref{lem_dgm_dec},
$\forall M \in \Omega(t+1)$,
$|x_i(M) - x_i(G)| \leq R(t)$.
By induction, $\forall t' > t$ and $\forall M \in \Omega(t')$,
$|x_i(M) - x_i(G)| \leq R(t)$.
In particular,
$|X_i(t) - x_i(G)| \leq R(t)$ and
$|X_i(t') - x_i(G)| \leq R(t)$.
Thus, $|X_i(t') - X_i(t)| \leq 2R(t)$.
\end{proof}

\begin{lemma}
\label{lemcauchy}
Let $(u_k)_k$ be a sequence,
Let $\alpha \in ]0,1[$ and let $N$ be an integer.
If $\forall k \geq N$,
$|u_{k+1} - u_k| \leq \alpha^k$,
then $(u_k)_k$ converges.
\end{lemma}

\begin{proof}As $\alpha \in ]0,1[$,
$S_{\alpha} = 1 + \alpha + \alpha^2 + \alpha^3 + \dots$ converges.
Let $\epsilon > 0$.
Let $K = \log ( \epsilon / S_{\alpha} ) / \log \alpha$
Then, $\alpha^K S_{\alpha} = \epsilon$.

Let $k \geq \max(K,N)$ and let $m > k$.
$|u_m - u_k| \leq \Sigma_{i = k}^{i = m-1} |u_{i+1} - u_i|
\leq \Sigma_{i = k}^{i = m-1} \alpha^i
\leq \alpha^k S_{\alpha}
\leq \alpha^K S_{\alpha} = \epsilon$

Thus, $(u_k)_k$ is a Cauchy sequence and it converges.\end{proof}

\begin{lemma}
\label{lem_bloc_2_2}
Let $\alpha \in ]0,1[$.
If, for any time $t$, there exists a time $t' > t$ such that $R(t') \leq \alpha R(t)$,
then the MM algorithm solves the convergence problem.
\end{lemma}

\begin{proof}
Let $t_0$ be an arbitrary time.
$\forall k \geq 0$,
we define $t_{k+1} > t_k$ as the first time such that $R(t_{k+1}) \leq \alpha R(t_k)$.
By induction, $\forall k \geq 0$,
$R(t_k) \leq \alpha^k R(t_0)$.

Let $i \in \{1,\dots,d\}$.
According to Lemma~\ref{lemshrink2rt},
$\forall k \geq 0$, we have
$|X_i(t_{k+1})-X_i(t_k)| \leq 2R(t_k) \leq 2 \alpha^k R(t_0)$.
$\forall k \geq 0$,
let $u_k = X_i(t_k)/(2R(t_0))$.
Then, $\forall k \geq 0$,
$|u_{k+1} - u_k| \leq \alpha^k$.

According to Lemma~\ref{lemcauchy},
the sequence $(u_k)_k$ converges and so does 
$(X_i(t_k))_k$.
Let $L_i$ be the limit of $(X_i(t_k))_k$,
and let $G$ be the point of coordinates
$(L_1, L_2,$ $\dots, L_d)$.

$R(t_k)$ decreases exponentially with $k$.
Then, $\forall \epsilon > 0$,
there exists an integer $k$ such that
$R(t_k) < \epsilon/2$.
According to Lemma~\ref{lem_r_is_dec},
$\forall t > t_k$,
$R(t) \geq R(t_k)$.
Therefore, the processes are
$(G,\epsilon)$-gathered $\forall t \geq t_k$,
and the convergence problem is solved.\end{proof}

\begin{lemma}
\label{lem_convv_poss}
If $d = 1$ or $n \leq 5$,
the MM algorithm solves the convergence problem.
\end{lemma}

\begin{proof}
If $d = 1$, according to Lemma~\ref{lemd1},
the MM algorithm solves the gathering problem,
and thus the convergence problem.
Now, suppose that $n \leq 5$.

Suppose that, for any time $t$, $|\Omega(t)| \geq 4$. Then, according to Lemma~\ref{lem_bloc_1_2} and Lemma~\ref{lem_bloc_2_2},
the MM algorithm solves the convergence problem.
Otherwise, i.e., if $|\Omega(t)| \leq 3$,
then according to Lemma~\ref{lem_of_3},
the MM algorithm solves the convergence problem.
\end{proof}

\begin{lemma}
\label{lem_convv_imposs}
If $d \geq 2$ and $n \geq 6$,
the convergence problem is impossible to solve.
\end{lemma}

\begin{proof}
Assume the opposite: there exists an algorithm that always solves the convergence problem
for $d \geq 2$ and $n \geq 6$.

First, assume that $\Omega$ contains $3$ points such as described in the proof of Lemma~\ref{lem3pt}.
Consider the infinite execution described in the proof of Lemma~\ref{lem3pt}.
Let $G$ be the barycenter of these $3$ points.

Let $P$ be the following proposition: there exists a constant $D$ such that the distance between $G$ and any of the $3$ points of $\Omega$ is at most $D$.

If $P$ is false, then by definition, the convergence problem cannot be solved. We now consider the case where $P$ is true.

If $P$ is true, then consider the following case:
$\Omega$ contains $6$ points
$K_1$, $K_2$, $K_3$, $K_4$, $K_5$ and $K_6$.
$K_1$, $K_2$ and $K_3$ are arranged such as described in the proof of
Lemma~\ref{lem3pt},
and so are $K_4$, $K_5$ and $K_6$.
Let $G$ (resp $G'$) be the barycenter of the triangle formed by $K_1$, $K_2$ and $K_3$
(resp. $K_4$, $K_5$ and $K_6$).
Assume that $d(G,G') = 10D$.

Now, assume that the points of the two triangles respectively follow the infinite execution described in the proof of Lemma~\ref{lem3pt}.
Then, the distance between any two of the $6$ points is always at least $8D$, and the convergence problem cannot be solved.\end{proof}

\textbf{Theorem~\ref{th_lb_conv}.} \emph{The convergence problem can be solved if and only if $d = 1$ or $n \leq 5$. When it can be solved, the MM algorithm solves it.}

\begin{proof}
The result follows from
Lemma~\ref{lem_convv_poss}
and
Lemma~\ref{lem_convv_imposs}.
\end{proof}

\section{Breaking symmetry}
\label{secpos}

We showed that the problems were impossible to solve for $n \geq 6$.
This is due to particular configurations where a process $p$ has several ``closest neighbors''
(i.e., $|N_p| > 1$).
Until now, we assumed that the actual closest neighbor $C(p)$ of $p$ was chosen in $N_p$ by an external adversary.

We now assume that, whenever $|N_p| > 1$, $C(p)$ is chosen deterministically,
according to an \emph{order} on the positions of processes.
Namely, we assume that there exists an order ``$<$'' such that any set of distinct points can be ordered from
``smallest'' to ``largest''
($A_1 < A_2 < A_3 < \dots < A_k$).



Let $L(p)$ be the largest element of $N_p$, that is:
$\forall q \in N_p - \{L(p)\}$,
$M_q < M_{L(p)}$.
We now assume that, for any process $p$, $C(p) = L(p)$.
With this new hypothesis, we show the following result.

\begin{theorem}
\label{th_mm}
$\forall n \geq 2$,
the MM algorithm solves the gathering problem in
$n - 1$ steps,
and no algorithm can solve the gathering problem in less that $n - 1$ steps.
\end{theorem}

\subsection*{Proof}

\begin{lemma}
\label{lem_lb_mm}
$\forall n \geq 2$,
no algorithm can solve the gathering problem in less than $n - 1$ steps.
\end{lemma}

\begin{proof}
Suppose the opposite: there exists an algorithm $X$ solving the gathering problem in less than $n - 1$ steps.

First, consider a case with two processes, initially at two distinct positions.
Then, eventually, the two processes are gathered.
Let $t$ be the first time where the two processes are gathered.
Let $A$ and $B$ be their position at time $t - 1$, and let $D = d(A,B)$.
By symmetry, the two processes should move to $m(A,B)$ at time $t$.
Thus, with algorithm $X$, whenever a process $p$ is such that
$d(M_p,M_{C(p)}) = D$,
$p$ moves to $m(M_p,M_{C(p)})$ at the next step.

Let $K(x)$ be the point of coordinates $(x,0,0,\dots,0)$.
Now consider $n$ processes, a set
$\Omega(0) = \bigcup_{i \in \{0,\dots,n-1\}}$ $\{K(iD)\}$, and an order such that,
$\forall x < y$, $K(x) < K(y)$.\footnote{As this is a lower bound proof, our goal here is to exhibit one particular situation where no algorithm can solve the problem in less than $n-1$ steps.
Thus, we choose a worst-case configuration with a worst-case order.}

Let us prove the following property $P_k$ by induction, $\forall k \in \{0,\dots,n-1\}$:
\\$\Omega(k) = \bigcup_{i \in \{0,\dots,n-k-1\}} \{K((i + k/2)D)\}$

\begin{itemize}

\item $P_0$ is true, as
$\Omega(0) = \bigcup_{i \in \{0,\dots,n-1\}} \{K(iD)\}$.

\item Suppose that $P_k$ is true for
$k \in \{0,\dots,n-2\}$.
Then, according to algorithm $X$,
the processes at position
$K((n-k-1 + k/2)D)$ moves to
$K((n-k-1 + (k-1)/2)D)$,
and $\forall i \in \{0,\dots,n-k-2\}$,
the processes at position
$K((i + k/2)D)$ move to
$K((i + (k+1)/2)D)$.
Thus, $P_{k+1}$ is true.\end{itemize}

Therefore, $\forall t \in \{0,\dots,n-2\}$,
$|\Omega(t)| \geq 2$, and the processes are not gathered: contradiction. Thus, the result follows.\end{proof}

We now assume that the processes move according to the MM algorithm.

\begin{lemma}
\label{lem_twoproc}
Let $p$ and $q$ be two processes.
If there exists a time $t$ where $M_p = M_q$,
then at any time $t' > t$, $M_p = M_q$.
\end{lemma}

\begin{proof}
Consider the configuration at time $t$.
According to our new hypothesis, $C(p) = C(q)$.
Let $K = m(M_p,M_{C(p)}) = m(M_q,M_{C(q)})$.
According to the MM algorithm, $p$ and $q$ both move to $K$.
Thus, at time $t+1$, we still have $M_p = M_q$. Thus, by induction, the result.\end{proof}

\begin{lemma}
\label{lem_notgaz}
At any time $t$, if the processes are not gathered, there exists two processes $p$ and $q$ such that
$M_p \neq M_q$, $p = C(q)$ and $q = C(p)$.
\end{lemma}

\begin{proof}
Let $\delta = \min_{\{A,B\} \subseteq \Omega(t)} d(A,B)$.
Let $Z$ be the set of processes $p$ such that $d(M_p,M_{C(p)}) = \delta$.
Let $Z' = \bigcup_{p \in Z } \{p, C(p)\}$.

Let $A$ be the point of $Z'$ such that,
$\forall M \in Z' - \{A\}$, $M < A$.
Let $p$ be a process at position $A$.

Let $q$ be the largest element of $N_p$, that is:
$\forall q' \in N_p - \{q\}$, $M_{q'} < M_q$.
By definition, $M_p \neq M_q$.
Thus, according to our new hypothesis, $q = C(p)$.

Then, note that $p$ is also the largest element of $N_q$:
$\forall p' \in N_q - \{p\}$, $M_{p'} < M_p$.
Thus, $p = C(q)$. Thus, the result follows.
\end{proof}

\begin{lemma}
\label{lem_indec}
At any time $t$, if the processes are not gathered, then
$|\Omega(t+1)| \leq |\Omega(t)| - 1$.
\end{lemma}

\begin{proof}
Let $p$ and $q$ be the processes described in
Lemma \ref{lem_notgaz}.
Let $K = m(M_p,M_q)$.
Then, according to Lemma~\ref{lem_notgaz},
the processes at position $M_p$ and $M_q$ both move to position $K$.
Let $X = \Omega(t) - \{M_p,M_q\}$.
According to
Lemma~\ref{lem_twoproc},
the processes occupying the positions of $X$ cannot move to more than $|X|$ new positions.
Thus, $|\Omega(t+1)|$ is at most $|\Omega(t)| - 1$. Thus, the result follows.\end{proof}

\begin{lemma}
\label{lem_pre_th_mm}
$\forall n \geq 2$,
the MM algorithm solves the gathering problem in at most $n-1$ steps.
\end{lemma}

\begin{proof}
According to Lemma~\ref{lem_indec},
there exists a time $t \leq n-1$
such that $|\Omega(t)| = 1$.
Let $A$ be the only point of $\Omega(t)$.
Then, according to the MM algorithm, the processes do not move from position $A$ in the following steps. Thus, the result follows.
\end{proof}

\textbf{Theorem~\ref{th_mm}}.
\textit{$\forall n \geq 2$, the MM algorithm solves the gathering problem in
$n - 1$ steps,
and no algorithm can solve the gathering problem in less that $n - 1$ steps.}

\begin{proof}
The result follows from Lemma~\ref{lem_lb_mm} and Lemma~\ref{lem_pre_th_mm}.
\end{proof}

\section{Fault tolerance}
\label{secft}

We now consider the case of \emph{crash failures}: some processes may lose the ability to move, without the others knowing it.
Let $C \subseteq P$ be the set of crashed processes (the other processes are called ``correct''), and let $S_c = \bigcup_{p \in C} \{M_p\}$
(i.e., the set of positions occupied by crashed processes). Let $f = |S_c|$.

We prove the two following results.

\begin{theorem}
\label{th_fault_1}
The gathering problem can only be solved when $f = 0$.
\end{theorem}

\begin{theorem}
\label{th_fault_2}
The convergence problem can be solved if and only if $f \leq 1$. When $f \leq 1$, the MM algorithm solves it.
\end{theorem}

\subsection*{Proof}

We say that a process $p$ is \emph{attracted} if there exists a sequence of processes $(p_1,\dots,p_m)$ such that $p = p_1$, $p_m \in C$, and $\forall i \in \{1,\dots,m-1\}$, $C(p_i) = p_{i+1}$.
A \emph{loop} is a sequence of correct processes
$(p_1,\dots,p_m)$ such that $C(p_m) = p_1$ and, $\forall i \in \{1,\dots,m-1\}$,
$C(p_i) = p_{i+1}$.
A \emph{pair} is a loop with $2$ processes.
Let $\Omega' = \bigcup_{p \in P - C} \{M_p\}$
(i.e., the set of positions occupied by \emph{correct} processes).
Let $\Omega'(t)$ be the state of $\Omega'$ at time $t$.

\begin{lemma}
\label{lemft0}
Consider an algorithm for which there exists $w$ such that $f_x(w) = w$ and $f_y(w) = 0$.
Then, this algorithm cannot solve the gathering nor the convergence problem.
\end{lemma}

\begin{proof}
Assume the opposite. Consider a situation where $\Omega = \{A,B\}$,
with $d(A,B) = w$.
Then, according to the algorithm, the processes at position $A$ and $B$ switch their positions endlessly, and neither converge nor gather: contradiction. Thus, the result follows.
\end{proof}

\textbf{Theorem~\ref{th_fault_1}.}
\emph{The gathering problem can only be solved when $f = 0$.}

\begin{proof}
If $f \geq 2$, by definition, the processes cannot be gathered. Now, suppose $f = 1$.

Suppose the opposite of the claim: there exists an algorithm solving the gathering problem when $f = 1$.
Let $P$ be the following proposition: there exists two points $A$ and $B$ such that all crashed processes are in position $A$, and all correct processes are in position $B$.

Consider an initial configuration where $P$ is true.
As the algorithm solves the gathering problem, according to Lemma~\ref{lemft0},
the next position of correct processes cannot be $A$.
Thus, $P$ is still true at the next time step, with a different point $B$.

Therefore, by induction, $P$ is always true, and the processes are never gathered: contradiction.
Thus, the result follows.
\end{proof}

\begin{lemma}
\label{lem_exloop}
If there exists a process $p$ which is not attracted, then there exists a loop.
\end{lemma}

\begin{proof}
Suppose the opposite: there is no loop. Let $p_1 = p$. $\forall i \in \{1,\dots,n\}$, let $p_{i+1} = C(p_i)$.

We prove the following property $P_i$ by induction,
$\forall i \in \{1,\dots,n+1\}$:
$(p_1,\dots,p_i)$ are $i$ distinct processes.

\begin{itemize}
\item $P_1$ is true.

\item Suppose that $P_i$ is true for some $i \in \{1,\dots,n\}$.
As there is no loop, we cannot have
$p_{i+1} \in \{p_1,\dots,p_i\}$.
Thus, $P_{i+1}$ is true.
\end{itemize}

Thus, $P_{n+1}$ is true, and there are $n+1$ distinct processes: contradiction.
Thus, the result follows.\end{proof}

\begin{lemma}
\label{lem_loopair}
All loops are pairs.
\end{lemma}

\begin{proof}
Let $(p_1,\dots,p_m)$ be a loop.
Let $\delta = \min_{i \in \{1,\dots,m\}}$ $d(M_{p_i},M_{C(p_i)})$.
Let $Z$ be the set of processes of $\{p_1,\dots,$ $p_m\}$ such that $d(M_{p_i},M_{C(p_i)}) = \delta$.
Let $Z' = \bigcup_{p \in Z}$ $\{p,C(p)\}$.

Let $p$ be the process such that, $\forall q \in Z'$ such that $M_p \neq M_q$, $M_p > M_q$.
Let $q = C(p)$.
As $C(q)$ is the closest neighbor of $p$, $C(q) \in Z'$.
Then, according to the definition of $p$, $C(q) = p$.

Therefore, $(p_1,\dots,p_m)$ is either $(p,q)$ or $(q,p)$. Thus, the result follows.
\end{proof}

\begin{lemma}
\label{lem_omdec}
If there exists a pair, then $|\Omega'(t+1)| \leq |\Omega'(t) - 1|$.
\end{lemma}

\begin{proof}
According to the algorithm, two processes at the same position at time $t$ are at the same position at time $t+1$.
Let $(p,q)$ be a pair. Then, according to the algorithm, the processes at positions $M_p$ and $M_q$ move to $m(M_p,M_q)$, and
$|\Omega'(t+1)| \leq |\Omega'(t) - 1|$.
\end{proof}

\begin{lemma}
\label{lem_t_a}
There exists a time $t_A$ such that, for any time $t \geq t_A$, all correct processes are attracted.
\end{lemma}

\begin{proof}
Suppose the opposite. Then, after a finite number of time steps, at least one correct process is not attracted.
Thus, according to Lemma~\ref{lem_exloop},
there exists a loop.
According to Lemma~\ref{lem_loopair},
this loop is a pair.
Then, according to Lemma~\ref{lem_omdec},
$|\Omega'|$ decreases.

We can repeat this reasoning $n + 1$ times, and we then have $|\Omega'| < 0$: contradiction. Thus, the result follows.
\end{proof}

\begin{lemma}
\label{lem_l_sur_n}
Suppose $f = 1$.
Let $p$ be an attracted process,
and let $L$ be the distance between $p$ and the crashed processes.
Then, $d(M_p,M_{C(p)}) \geq L/n$.
\end{lemma}

\begin{proof}

Suppose the opposite:
$d(M_p,M_{C(p)}) < L/n$.
As $p$ is attracted, there exists
a sequence of processes $(p_1,\dots,$ $p_m)$ such that $p = p_1$, $p_m \in C$, and $\forall i \in \{1,\dots,m-1\}$, $C(p_i) = p_{i+1}$.

$\forall i \in \{1,\dots,m-2\}$, we have
$d(M_{p_i}, M_{p_{i+1}}) \geq$ \\$d(M_{p_{i+1}}, M_{p_{i+2}})$.
Indeed, suppose the opposite. 
Then, $C(p_{i+1}) = p_{i+2}$,
$d(M_{p_{i+1}}, M_{C(p_{i+1})}) > d(M_{p_{i+1}},M_{p_i})$,
and $C(p_{i+1})$ is not a closest neighbor of $p_{i+1}$: contradiction.
Thus, $d(M_{p_i}, M_{p_{i+1}}) \geq d(M_{p_{i+1}},M_{p_{i+2}})$.

Thus, $\forall i \in \{1,\dots,m-1\}$,
$d(p_i,p_{i+1}) < L/n$.
Therefore, $d(p_1,p_m) \leq (m-1)L/n < L$:
contradiction. Thus, the result follows.
\end{proof}

\begin{lemma}
\label{lem_k_n}
Let $f = 1$, and let $X$ be the position of crashed processes.
Let $L = \max_{p \in P} d(X,M_p)$.
Let $L(t)$ be the value of $L$ at time $t$.
Suppose that all correct processes are attracted.
Then, for any time $t$, $L(t+1) \leq k(n)L(t)$,
where $k(n) = \sqrt{1 - 1/(2n)^2}$.

\end{lemma}

\begin{proof}
At time $t + 1$, let $p$ be a process such that $d(X,M_p)$ $= L(t+1)$.
Let $K$ be the position of $p$ at $t+1$.
Then, according to the algorithm, at time $t$,
there exists two processes $q$ and $r$ at position $A$ and $B$
such that $K = m(A,B)$.

Let $L' = \max(d(X,A),d(X,B))$.
Let $q' \in \{q,r\}$ be such that
$d(X,M_{q'}) = L'$.
Then, according to Lemma~\ref{lem_l_sur_n},
$d(M_{q'},M_{C(q')}) \geq L'/n$.
Let $r'$ be the other process of $\{q,r\}$. Then, the position of $r$ maximizing $d(X,K)$ is such that
$d(X,M_{r'}) = L'$.

Therefore, according to the Pythagorean theorem, $(L(t+1))^2$ is at most
$L'^2 - (L'/(2n))^2$,
and $L(t+1) \leq k(n)L' \leq k(n)L(t)$.
Thus, the result follows.
\end{proof}

\begin{lemma}
\label{lemconv_f1}
If $f = 1$, the MM algorithm solves the convergence problem.
\end{lemma}

\begin{proof}

According to Lemma~\ref{lem_t_a},
there exists a time $t_A$ after which all correct processes are attracted.
We now suppose that $t \geq t_A$.

Let $\epsilon > 0$.
Let $X$ be the position of crashed processes, and let
$L = \max_{p \in P} d(X,M_p)$.
As $k(n) =$ \\ $\sqrt{1 - 1/(2n)^2} < 1$,
let $M$ be such that
$k(n)^M L < \epsilon$.
Then, according to
Lemma~\ref{lem_k_n},
at time $t_A + M$, all processes are at distance at most $\epsilon$ from $X$. Thus, the result follows.
\end{proof}

\textbf{Theorem~\ref{th_fault_2}.}
\emph{The convergence problem can be solved if and only if $f \leq 1$. When $f \leq 1$, the MM algorithm solves it.}

\begin{proof}
When $f \geq 2$,
there exists at least two crashed processes that will stay at the same position forever. Thus, the convergence problem cannot be solved. 

When $f \leq 1$, according to Lemma~\ref{lemconv_f1},
the MM algorithm solves the convergence problem. Thus, the result follows.
\end{proof}

\section{Conclusion}
\label{sec_conc}

In this paper, we revisited the gathering and convergence problems with a minimalist hypothesis on visibility. We showed that this model only allows a small number of processes to converge, but requires an additional symmetry-breaking hypothesis to gather an arbitrarily large number of processes.
For the convergence problem, up to one crash failure can be tolerated.

This first work can be the basis for many extensions. For instance, we could consider a more general scheduler (e.g. asynchronous). We could investigate how resilient this model is to crash or Byzantine failures. We could also consider the case of voluminous processes, that cannot be reduced to one geometrical point.

\section*{Acknowledgments}

This work has been supported by the Swiss National Science Foundation
(Grant 200021\_169588 TARBDA).

\bibliographystyle{plain}

\end{document}